\documentclass[twocolumn]{autart}

\usepackage{graphicx}          

\usepackage{amsmath,amssymb,amsfonts,mathrsfs,bm}
\usepackage{enumitem}
\usepackage{enumitem}
\usepackage{threeparttable}
\usepackage{booktabs}
\usepackage{multirow}
\usepackage{xcolor}
\newif\ifrevmarks \revmarksfalse 
\ifrevmarks
  \newcommand{\rev}[1]{{\color{blue}#1}}
  \newcommand{\revcolor}{\color{blue}}
\else
  \newcommand{\rev}[1]{{#1}}
  \newcommand{\revcolor}{}
\fi
\makeatletter 
\def\@endtheorem{\ifhmode\unskip\ifnum\lastnodetype=\m@ne\par
  \@tempskipa\lastskip \advance\@tempskipa-\parskip \vskip-\lastskip
  \ifdim\@tempskipa<\z@ \@tempskipa\z@\fi \vskip\@tempskipa\fi\fi
  \endtrivlist}
\makeatother

\newcommand{\la}{\langle}
\newcommand{\ra}{\rangle}

\newcommand{\mrm}[1] {\mathrm{#1}}
\newcommand{\mbf}[1] {\mathbf{#1}}

\newcommand{\re}[1] {\mathrm{Re}\{#1\}}
\newcommand{\im}[1] {\mathrm{Im}\{#1\}}
\newcommand{\diag}[1] {\mathrm{diag}\{#1\}}

\newcommand{\ml}{\mathrm{l}}

\newcommand{\mP}{\mathrm{P}}
\newcommand{\mI}{\mathrm{I}}
\newcommand{\md}{\mathrm{d}}
\newcommand{\mr}{\mathrm{r}}
\newcommand{\mpv}{\mathrm{pv}}

\newcommand{\HT}{\mathsf{H}}
\newcommand{\calG}{\mathcal{G}}

\newcommand{\calY}{\mathcal{Y}}
\newcommand{\calD}{\mathcal{D}}

\newcommand{\calDcap}{\mathcal{D}_\cap}

\begin{document}
\begin{frontmatter}
\title{\rev{Decentralized Certification of Regional Pole Placement via Positive Functions with Application to DC Microgrids}}


\author[DA]{Zelin Sun}\ead{szl24@mails.tsinghua.edu.cn},    
\author[DA]{Shanshan Jiang}\ead{jss23@mails.tsinghua.edu.cn},
\author[EE]{Xiaoyu Peng}\ead{pengxy19@tsinghua.org.cn},
\author[DA]{Xiang Zhu}\ead{zhu-x22@mails.tsinghua.edu.cn},
\author[DA]{Xiuqiang He}\ead{xhe@tsinghua.edu.cn},
\author[DA]{Hua Geng}\ead{genghua@tsinghua.edu.cn}

\address[DA]{Department of Automation, Tsinghua University, Beijing, China}    
\address[EE]{Department of Electrical Engineering, Tsinghua University, Beijing, China}

\begin{keyword}
Regional pole placement, decentralized certification, positive function, positive realness, DC microgrid, grid code.
\end{keyword}

\begin{abstract}
This paper proposes a decentralized method for certifying regional pole placement in linear time-invariant networked systems, i.e., placing all closed-loop poles in a target region. Each subsystem model is known only to its owner, and the network only to a system operator, so that no single party can compute the closed-loop poles. 
Existing methods based on linear matrix inequalities require all subsystem models or communication links among subsystems, which may be absent. 
To avoid both requirements, we extend the compositional framework of positive realness (PR) from stability to regional pole placement. A complex mapping turns regional pole placement into stability in an auxiliary plane, and an angle compensation for each subsystem removes its mapping-induced phase lag. 
For the resulting complex-coefficient subsystems, we prove that the PR of their real-equivalent representations is equivalent to their positivity. This yields a certificate consisting of a positivity condition on each subsystem, checked by its owner, and a condition involving only the network, checked by the operator. For DC microgrids, the certificate is cast as a grid code that remains valid over a prescribed set of operating points, so that changes in generation and demand within this set require no update. A case study verifies the method.
\end{abstract}

\end{frontmatter}

\section{Introduction}
The dynamic performance of a linear time-invariant networked system is largely determined by the location of its closed-loop poles. Stability requires all poles to reside in the closed left half-plane (cLHP). Practical performance specifications, such as a minimum decay rate or damping ratio, further confine the poles to specific regions of the complex plane \cite{feng2025hybrid}. Such requirements are known as \textit{regional pole placement} \cite{chilali1996h}.

Regional pole placement is traditionally achieved via linear matrix inequality (LMI) methods \cite{chilali1996h,scherer1997multiobjective,chilali1999robust,duan2013lmis}. These methods characterize regional pole constraints by LMIs in the closed-loop state matrix and a Lyapunov matrix, formulating full-order output-feedback synthesis as a convex optimization problem.
However, \rev{they are centralized and require the state-space model of the entire system. In many networked systems, this model is split among several parties: each subsystem model is known to its owner but often proprietary, whereas the network is known to a system operator} \cite{peng2026compositional}. \rev{Since the closed-loop poles depend jointly on all subsystem models and the network, no single party can compute or assign them.} The recent distributed LMI design \cite{koru2022regional} enhances scalability by synthesizing the gains agent by agent. Its control law, however, uses the relative outputs of neighboring subsystems and thus requires information exchange among them. In practical infrastructures, this exchange relies on communication links, which may be absent or unreliable. Therefore, there is a need for a \textit{decentralized} certificate of regional pole placement, one that requires neither shared subsystem models nor communication among the subsystems. Such a certificate also lets each owner design its controller independently.

\rev{Decentralized certificates of this kind, which combine a property of each subsystem with a condition on the interconnection, date back to stability criteria for large-scale systems \cite{moylan1978stability} and are a focus of recent research on power grids \cite{laib2023decentralized,huang2024gain,dey2022passivity,he2024passivity,peng2026compositional,yang2019distributed}.} A classical framework for them is positive realness (PR) theory, the frequency-domain counterpart of passivity \cite{kottenstette2014relationships}. The cornerstone of PR is its compositional property: the negative feedback interconnection of PR operators preserves the PR property. Hence, system stability can be guaranteed by local frequency-domain compliance \cite{brogliato2020dissipative,peng2026compositional}. However, PR conditions are stated on the imaginary axis and thus certify only that the closed-loop poles lie in the cLHP, not in a region that enforces, e.g., a minimum decay rate or damping ratio. Extending this compositional framework to regional pole placement remains an open challenge.

DC microgrids (DCMGs) illustrate why this gap matters. In a DCMG, rapid voltage recovery and the suppression of persistent oscillations require regional pole placement \cite{hamzeh2016power}. Yet, its devices, such as photovoltaic arrays (PVs), energy storage systems (ESSs), and constant power loads (CPLs), often belong to different stakeholders. Their vendors are reluctant to share detailed device models, and real-time communication links among the devices can be absent. A decentralized certificate can be cast as a ``grid code'' \cite{peng2026compositional,sun2026passivity}. The system operator specifies the target region of the closed-loop poles and a requirement computed from the network, and each device complies with them using only its own model. This supports plug-and-play operation.

\rev{This paper extends the compositional PR framework from stability to regional pole placement. A complex mapping turns regional pole placement into stability in an auxiliary plane. 
This mapping, however, introduces phase lags that generally violate the PR conditions, and it yields subsystems with complex coefficients, to which PR theory does not apply directly. We overcome both obstacles. 
An angle compensation for each subsystem removes its phase lag.
As our primary theoretical result, we prove that the PR of a mapped subsystem's real-equivalent representation is equivalent to the positivity of its complex-coefficient transfer function (Lem.~\ref{lem:positive function}), so that the conditions can be checked for the complex-coefficient subsystems themselves. 
Building on this result, Th.~\ref{thm: D stable with positivity} certifies regional pole placement in a decentralized manner by a local positivity condition on each subsystem and a condition on the network, thereby bringing positive functions from classical network theory \cite{belevitch1968classical} to this problem.
Each owner checks the local condition of its subsystem using its own model, and the system operator checks the network condition. 
For DCMGs, these conditions are cast as a grid code that remains valid over a prescribed set of operating points. For second-order subsystems, the local conditions are ensured by closed-form constraints, allowing explicit tuning.}

\textit{Notation.} $\mathbb{R}$ and $\mathbb{C}$ denote the sets of real and complex numbers, and $\mathbb{R}^{n \times m}(s)$ and $\mathbb{C}^{n \times m}(s)$ the sets of $n \times m$ rational transfer function matrices with real and complex coefficients. $I_n$ denotes the $n \times n$ identity matrix; the subscript is omitted when the size is clear. For $s \in \mathbb{C}$, $\re{s}$ and $\im{s}$ denote its real and imaginary parts. For a matrix $A$, $A^\top$ and $A^\HT$ denote its transpose and Hermitian transpose. For a Hermitian matrix $A$, $A \succ 0$ ($A \succeq 0$) means that $A$ is positive (semi-)definite, and $\lambda_\mrm{min}(A)$ denotes its smallest eigenvalue. $A \otimes B$ is the Kronecker product, and $\diag{A_1, \dots, A_N}$ denotes the block-diagonal matrix with diagonal blocks $A_k$.

\section{Problem Formulation}\label{section: theory}
Consider a network of $N$ linear time-invariant subsystems
\begin{equation}\label{eq:subsystem}
\eta_k(s)= G_k(s) w_k(s), \quad k = 1, \dots, N,
\end{equation}
where the proper transfer function matrix $G_k(s) \in \mathbb{R}^{m \times m}(s)$ maps the input $w_k$ to the output $\eta_k$. The subsystems interact through the static network
\begin{equation}\label{eq:network}
w = -Y \eta,
\end{equation}
where $w = [w_1^\top, \dots, w_N^\top]^\top$, $\eta = [\eta_1^\top, \dots, \eta_N^\top]^\top$, and the matrix $Y \in \mathbb{R}^{mN \times mN}$ encodes the topology and the coupling weights. With $G(s) = \diag{G_1(s), \dots, G_N(s)}$, the closed-loop poles are the roots of
\begin{equation}\label{eq:char_eq}
\det(I + G(s)Y) = 0 .
\end{equation}

\begin{defn}\label{def:regional}
\rev{(Regional pole placement \cite{chilali1996h})} Given a region $\calD\subset\mathbb{C}$, a dynamical system is $\calD$-stable if all its poles lie in $\mathcal{D}$. \rev{Throughout this paper, $\calD$-stability refers to this property. It should not be confused with the D-stability of matrices \cite{arrow1958note}, which concerns the robustness of matrix stability against positive diagonal scaling.}
\end{defn}

The target regions are built from the rotated and shifted half-plane in Fig.~\ref{fg:half plane}(a):
\begin{equation}\label{eq:calD_0}
\mathcal{D}_0(\theta_0, \omega_0, \sigma_0) \triangleq \{ s \mid \re{e^{-j\theta_0}(s - j\omega_0)} \leq \sigma_0 \}
\end{equation}
Since the system \eqref{eq:subsystem}-\eqref{eq:network} is real rational, its closed-loop poles appear in conjugate pairs. As $\mathcal{D}_0(-\theta_0, -\omega_0, \sigma_0)$ is the mirror image of $\mathcal{D}_0(\theta_0, \omega_0, \sigma_0)$ about the real axis, all poles lie in $\mathcal{D}_0(\theta_0, \omega_0, \sigma_0)$ if and only if they lie in the symmetric intersection in Fig.~\ref{fg:half plane}(a):
\begin{equation}\label{eq:simple symmetric region}
\mathcal{D}(\theta_0, \omega_0, \sigma_0) \triangleq \mathcal{D}_0(\theta_0, \omega_0, \sigma_0) \cap \mathcal{D}_0(-\theta_0, -\omega_0, \sigma_0)
\end{equation}
Special choices of $(\theta_0, \omega_0, \sigma_0)$ yield the performance regions in Fig.~\ref{fg:half plane}(b)--(d):
\begin{equation}\label{eq:performance regions}
\!\!\begin{aligned}
&\calD_\mrm{LHP}(\alpha)\triangleq\calD(0,0,\alpha),~~ \calD_\mrm{SEC}(\beta)\triangleq\calD(\tfrac{\pi}{2}-\beta,0,0),\\
&\calD_\mrm{HS}(\gamma)\triangleq\calD(\tfrac{\pi}{2},\gamma,0),
\end{aligned}
\end{equation}
with $\alpha\le0$, $\beta\in(0,\frac{\pi}{2})$, and $\gamma>0$. The shifted left half-plane $\calD_\mrm{LHP}(\alpha)$ enforces a minimum decay rate $|\alpha|$, linked to the settling time $T_s\approx4/|\alpha|$, and equals the cLHP for $\alpha=0$. The sector $\calD_\mrm{SEC}(\beta)$, centered on the negative real axis, enforces a minimum damping ratio $\cos\beta$. The horizontal strip $\calD_\mrm{HS}(\gamma)$ enforces $|\im{s}|\le\gamma$, i.e., it limits the oscillation frequency to $\gamma$ rad/s.


\rev{As motivated in the Introduction, the network matrix $Y$ is known to a system operator, e.g., from the topology and the line parameters of an electrical network, whereas each $G_k(s)$ is known only to the owner of subsystem $k$, so that no single party can compute the closed-loop poles from \eqref{eq:char_eq}. Hence, $\mathcal{D}$-stability should be certified in a decentralized manner, i.e., by conditions each of which is verified by a single party with its own information}.

\begin{prob}\label{prob:main}
(Decentralized certification of regional pole placement) \rev{Derive a sufficient condition for $\mathcal{D}$-stability that consists of a local condition for each subsystem $k$ involving only $G_k(s)$, and a network condition involving only $Y$.}
\end{prob}

\rev{\begin{rem}\label{rem:local eigenvalue}
\textit{(Why the poles of the subsystems do not suffice)}
Placing the poles of each $G_k(s)$ in $\calD$, which each owner can do with its own model, is neither necessary nor sufficient for $\calD$-stability of the interconnection. In Section~\ref{section: case study}, an isolated CPL has an unstable pole, whereas the interconnected DCMG is $\calD$-stable. Conversely, with the buck ESS gains $(k_\mP^u,k_\mI^u)=(0.131,1)$ and the tuned boost ESS gains, every ESS has its poles in the target region, but the closed loop has a pole at $-0.66$.
\end{rem}}

\begin{figure}[t]
    \centering
\includegraphics{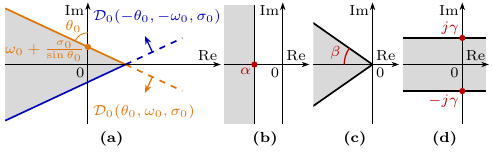}
    \caption{Target regions (gray) in the $s$-plane.}
    \label{fg:half plane}
\end{figure}

\section{Decentralized Certification of Regional Pole Placement}
\subsection{Region Mapping and Real-Equivalent Representation}

To reduce $\calD$-stability to stability, we use the mapping
\begin{equation}\label{eq:mapping}
\nu \mapsto s(\nu) = e^{j\theta_0}(\nu + \sigma_0) + j\omega_0 ,  
\end{equation}
which maps the cLHP of an auxiliary $\nu$-plane bijectively onto the half-plane $\calD_0(\theta_0,\omega_0,\sigma_0)$ of the $s$-plane (Fig.~\ref{fg:affine mapping}).

\begin{figure}[t]
    \centering
\includegraphics{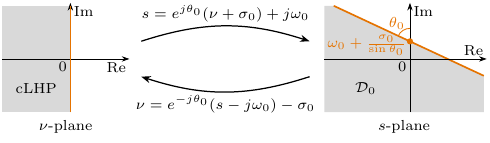}
    \caption{Mapping \eqref{eq:mapping} of the cLHP onto $\mathcal{D}_0(\theta_0,\omega_0,\sigma_0)$.}
    \label{fg:affine mapping}
\end{figure}

\begin{prop}\label{prop:transformation}
The interconnected system~\eqref{eq:subsystem}-\eqref{eq:network} is $\mathcal{D}$-stable w.r.t. $\mathcal{D}(\theta_0, \omega_0, \sigma_0)$ if and only if all roots of $\det\big(I+G(s(\nu))Y\big)=0$ satisfy $\re{\nu} \leq 0$.
\end{prop}
\textit{Proof:} By \eqref{eq:mapping}, $\nu=e^{-j\theta_0}(s-j\omega_0)-\sigma_0$, so $\re{\nu}\le0$ if and only if $s\in\calD_0(\theta_0,\omega_0,\sigma_0)$ by \eqref{eq:calD_0}. As shown before \eqref{eq:simple symmetric region}, all poles lie in $\calD_0$ if and only if they lie in $\calD$.\hfill$\square$

\rev{\textit{Interpretation of the mapping.} Since the inverse map $s\mapsto\nu=e^{-j\theta_0}(s-j\omega_0)-\sigma_0$ only rotates and shifts the plane and maps $\partial\calD_0$ onto the imaginary axis, $-\re{\nu}$ is the signed distance from $s$ to $\partial\calD_0$, positive inside $\calD_0$. For a pole $s_p$ in the upper half-plane, this distance is the decay-rate margin $\alpha-\re{s_p}$ for $\calD_\mrm{LHP}(\alpha)$, the damping margin $|s_p|\sin(\beta-\psi_p)$ for $\calD_\mrm{SEC}(\beta)$, where $\cos\psi_p$ is the damping ratio of $s_p$, and the frequency margin $\gamma-\im{s_p}$ for $\calD_\mrm{HS}(\gamma)$; the conjugate pole is covered by symmetry. Accordingly, the frequency responses in the $\nu$-plane used below are the responses of $G_k(s)$ along $\partial\calD_0$.}

\begin{rem}\label{rem:intersection}
The analysis extends directly to the intersection \rev{$\calDcap=\cap_{i=1}^n\mathcal{D}_i(\theta_{0,i}, \omega_{0,i}, \sigma_{0,i})$} of several sub-regions: \rev{$\calDcap$-stability} holds if and only if the roots of all $n$ mapped characteristic equations $\det\big(I+G(s_i(\nu))Y\big)=0$ satisfy $\re{\nu}\leq 0$, where $s_i(\nu)$ is \eqref{eq:mapping} with the parameters of $\calD_i$. Hence, it suffices to consider a single region $\calD$ in \eqref{eq:simple symmetric region}.
\end{rem}

In the $\nu$-plane, $\calD$-stability thus becomes stability, which PR theory certifies compositionally \cite{brogliato2020dissipative}.

\begin{defn}\label{def: passivity}\cite{anderson2013network,brogliato2020dissipative} A real-rational proper transfer function $M(s) \in \mathbb{R}^{m\times m}(s)$ is PR if and only if:
\begin{enumerate}[label={\alph*)}]
\item All poles of $M(s)$ are in $\re{s}\leq0$.
\item $M(j\omega)+M^\top(-j\omega)\succeq 0$ for all real $\omega$ for which $j\omega$ is not a pole of any element of $M(s)$.
\item Any pure imaginary pole $j\omega$ of any element of $M(s)$ is a simple pole, and the residue $\lim_{s\rightarrow j\omega}(s-j\omega)M(s)\succeq 0$.
\end{enumerate}
\end{defn}


However, PR criteria cannot be applied to the mapped subsystems directly. 
The mapping \eqref{eq:mapping} introduces a phase lag of $\theta_0$, which generally drives the frequency response of $G_k(s(\nu))$ outside the PR sector $[-\pi/2, \pi/2]$. 
For instance, if $G_k(s)$ has relative degree one with high-frequency gain $K_k \succ 0$ (i.e., $\lim_{s \to \infty} sG_k(s) = K_k$), we have $\lim_{\nu \to \infty} \nu G_k(s(\nu)) = e^{-j\theta_0} K_k$, which rotates its high-frequency response outside this sector.
To counteract this phase lag, we introduce a diagonal \textit{angle compensation matrix} $\Phi = \diag{\phi_1, \dots, \phi_N}$ and insert the identity $e^{j\Phi \otimes I_m} e^{-j\Phi \otimes I_m} = I$ into the mapped characteristic equation:
\begin{equation} \label{eq:rotated_det}
\det\Big(I + \underbrace{e^{j\Phi \otimes I_m} G(s(\nu))}_{\hat{G}(\nu)} \cdot \underbrace{ e^{-j\Phi \otimes I_m} Y}_{\hat{Y}}\Big) = 0
\end{equation}
where $\hat{G}(\nu) = \diag{\hat{G}_1(\nu), \dots, \hat{G}_N(\nu)}$ collects the rotated subsystems $\hat{G}_k(\nu) = e^{j\phi_k} G_k(s(\nu))$, and $\hat{Y}$ is the rotated network matrix. \rev{The choice of $\phi_k$ is discussed in Rem.~\ref{rem:theorem conditions}.}



The rotated subsystems $\hat{G}_k(\nu)$, however, still have complex coefficients, so PR for real-rational functions does not apply directly. To establish a theoretical connection to PR, for $M\in\mathbb{C}^{m\times m}(\nu)$, write $M=M_\mrm{re}+jM_\mrm{im}$ with $M_\mrm{re},M_\mrm{im}\in\mathbb{R}^{m\times m}(\nu)$, and define its real-equivalent $\la M\ra\in\mathbb{R}^{2m\times 2m}(\nu)$ by
\begin{equation}\label{eq:real-equivalent operator}
\la M\ra \triangleq
\begin{bmatrix}
M_\mrm{re} & -M_\mrm{im}
\\ M_\mrm{im} & M_\mrm{re}
\end{bmatrix}
=I_2\otimes M_\mrm{re}+J\otimes M_\mrm{im}
\end{equation}
with $J=\left[\begin{smallmatrix}0&-1\\1&0\end{smallmatrix}\right]$. Applying $\la\cdot\ra$ blockwise to \eqref{eq:rotated_det} yields the real-coefficient equation $\det(I + \calG(\nu)\calY) = 0$, where $\calG(\nu) = \diag{\la\hat{G}_1\ra(\nu), \dots, \la\hat{G}_N\ra(\nu)}$, and $\calY$ is formed by replacing each $m\times m$ block $\hat Y_{kl}$ with $\la\hat Y_{kl}\ra$. PR results thus apply to the real-equivalents $\la\hat G_k\ra$ and $\calY$, which are permutation-similar to $\la\hat{G}\ra$ and $\la\hat{Y}\ra$, preserving the closed-loop poles. 



\subsection{Main Results}\label{subsec:main result}

The following lemma expresses the PR of the real-equivalent $\la\hat G_k\ra$ in terms of $\hat G_k(\nu)$ itself. Indeed, $\la\hat G_k\ra$ is unitarily similar to $\diag{\check G_k,\hat G_k}$, where $\check G_k$ is $\hat G_k$ with conjugated coefficients and carries the identical information (see the proof below). The conditions of this lemma can be checked via LMIs or graphically (Rem.~\ref{rem:KYP}), and for second-order subsystems, they reduce to closed-form constraints (Prop.~\ref{prop:second-order coefficient condition}).

\begin{lem}\label{lem:positive function}
The real-equivalent $\la\hat G_k\ra\in\mathbb{R}^{2m\times 2m}(\nu)$ is PR if and only if the proper complex-coefficient $\hat{G}_k\in \mathbb{C}^{m\times m}(\nu)$ is a \textbf{positive function}, i.e., it satisfies:
\begin{enumerate}[label=\alph*)]
\item All poles of $\hat G_k(\nu)$ are in $\re{\nu}\leq0$.
\item $\hat G_k(j\omega)+\hat G_k(j\omega)^\HT\succeq 0$ for all real $\omega$ for which $j\omega$ is not a pole of any element of $\hat G_k(\nu)$.
\item Any pure imaginary pole $j\omega$ of any element of $\hat G_k(\nu)$ is a simple pole, and the residue $\lim_{\nu\rightarrow j\omega}(\nu-j\omega)\hat G_k(\nu)\succeq 0$.
\end{enumerate}
\end{lem}

\textit{Proof:} 
We fix $k$ and write $\hat G$ for $\hat G_k$. A square proper rational matrix $M$ is called positive if it satisfies conditions \mbox{a)--c)} of the lemma. Write $\hat G=\hat G_\mrm{re}+j\hat G_\mrm{im}$ with real-rational $\hat G_\mrm{re}$ and $\hat G_\mrm{im}$, and let $\check G\triangleq\hat G_\mrm{re}-j\hat G_\mrm{im}$. We prove the chain: $\la\hat G\ra$ is PR $\overset{\text{(i)}}{\Leftrightarrow}$ $\la\hat G\ra$ is positive $\overset{\text{(ii)}}{\Leftrightarrow}$ $\diag{\check G,\hat G}$ is positive $\overset{\text{(iii)}}{\Leftrightarrow}$ $\check G$ and $\hat G$ are positive $\overset{\text{(iv)}}{\Leftrightarrow}$ $\hat G$ is positive.

(i) For a real-rational $M$, $M(-j\omega)$ is the complex conjugate of $M(j\omega)$, so that $M^\top(-j\omega)=M(j\omega)^\HT$. Hence, Def.~\ref{def: passivity} coincides with \mbox{a)--c)}, and PR and positivity are the same for the real-rational and proper $\la\hat G\ra$.

(ii) Let $U=Q\otimes I_m$ with the unitary matrix $Q=\frac{1}{\sqrt2}\left[\begin{smallmatrix}1&1\\j&-j\end{smallmatrix}\right]$. Since $Q^\HT Q=I_2$ and $Q^\HT JQ=\diag{-j,j}$, \eqref{eq:real-equivalent operator} yields
\begin{equation}\nonumber
U^\HT\la\hat G\ra U=\diag{\hat G_\mrm{re}-j\hat G_\mrm{im},\,\hat G_\mrm{re}+j\hat G_\mrm{im}}=\diag{\check G,\hat G}.
\end{equation}
This constant unitary similarity preserves positivity. It preserves the poles and their orders, and it maps $M(j\omega)+M(j\omega)^\HT$ and the residues to $U^\HT(\cdot)U$, which preserves positive semi-definiteness.

(iii) A block-diagonal matrix is positive if and only if each block is, because its poles are those of the blocks, and its Hermitian part and residues are block-diagonal. One detail concerns b): the block-diagonal matrix is tested at frequencies at which neither block has a pole. At the finitely many frequencies at which only the other block has a pole, b) for a given block follows by continuity.

(iv) Since $\check G$ is $\hat G$ with all coefficients conjugated, its poles are the conjugates of those of $\hat G$, with the same real parts and orders. The residue of $\check G$ at $j\omega$ is the conjugate of the residue of $\hat G$ at $-j\omega$, and $\check G(j\omega)+\check G(j\omega)^\HT$ is the conjugate of $\hat G(-j\omega)+\hat G(-j\omega)^\HT$. Since conjugation preserves the eigenvalues of Hermitian matrices, $\check G$ is positive if and only if $\hat G$ is.\hfill$\square$

\begin{rem}\label{rem:existing vs new}
\rev{\textit{(Existing concepts and the new link)}} The conditions in Lem.~\ref{lem:positive function} resemble the classical PR conditions in Def.~\ref{def: passivity}. However, PR requires $M(\nu)$ to be real for real $\nu$ (i.e., to have real coefficients) \cite{brune1931synthesis,anderson2013network,brogliato2020dissipative}, which the complex-coefficient $\hat G_k(\nu)$ violates. Functions satisfying the conditions of Lem.~\ref{lem:positive function} are the positive functions in classical network theory \cite{belevitch1968classical}. \rev{The contribution of Lem.~\ref{lem:positive function} is the explicit equivalence between the PR of the real-equivalent $\la\hat G_k\ra(\nu)$ and the positivity of the complex-coefficient $\hat G_k(\nu)$.} To the best of the authors' knowledge, no proof of this equivalence is available in the literature. \rev{This equivalence enables the PR conditions to be formulated and checked for the complex-coefficient subsystems themselves. Th.~\ref{thm: D stable with positivity} builds on this result; so do the graphical test of Rem.~\ref{rem:KYP} and the closed-form constraints of Prop.~\ref{prop:second-order coefficient condition}.}
\end{rem}

The following theorem extends the use of positive functions to the decentralized certification of regional pole placement.



\begin{thm} \label{thm: D stable with positivity}
System \eqref{eq:subsystem}-\eqref{eq:network} is $\mathcal{D}$-stable w.r.t. the region $\mathcal{D}(\theta_0, \omega_0, \sigma_0)$ if there exists an angle compensation matrix $\Phi$ such that the rotated network matrix satisfies $\hat Y + \hat Y^\HT\succeq 0$ and each rotated subsystem $\hat G_k(\nu)$ is a positive function.
\end{thm}
\textit{Proof:} By Lem.~\ref{lem:positive function}, $\hat Y+\hat Y^\HT\succeq0$ implies that $\la\hat Y\ra$, and thus the permutation-similar $\calY$, is PR. Likewise, each $\la\hat G_k\ra(\nu)$, and thus the block-diagonal $\calG(\nu)$, is PR. As the negative feedback interconnection of PR operators is PR \cite{brogliato2020dissipative}, all roots of $\det(I+\calG(\nu)\calY)=0$ lie in the cLHP. By the unitary similarity in the proof of Lem.~\ref{lem:positive function}, they are the roots of \eqref{eq:rotated_det} and their complex conjugates, which have the same real parts. The claim follows from Prop.~\ref{prop:transformation}.\hfill$\square$

\rev{\begin{rem}\label{rem:theorem conditions}
(i) \textit{Sufficiency.} Th.~\ref{thm: D stable with positivity} gives a sufficient but not necessary condition. As in classical PR-based conditions, its conservatism arises because the subsystems and the network are required to be positive separately.
(ii) \textit{Choice of $\Phi$.} By Lem.~\ref{lem:positive function}, a proper positive function, like a PR function, has relative degree zero or one. If $G_k(s)$ has relative degree one with $\lim_{s\to\infty}sG_k(s)=K_k\succ0$, positivity of $\hat G_k$ requires the limit $\lim_{\nu\to\infty}\nu\hat G_k(\nu)=e^{j(\phi_k-\theta_0)}K_k\succeq 0$, which forces $\phi_k=\theta_0$. 
Hence, fixing $\phi_k=\theta_0$ adds no conservatism beyond that in (i). 
Relative degree one with $K_k\succ0$ is typical for physical ports in power grids, such as the DCMG devices with a shunt capacitor in Section~\ref{section: Application to DCMG}, so that $\Phi=\theta_0 I$ is fixed by the target region. 
Moreover, a simple pole $p\in\partial\calD_0$ of $G_k(s)$ with residue $r$ becomes the imaginary-axis pole $\nu_p\triangleq e^{-j\theta_0}(p-j\omega_0)-\sigma_0$ of $\hat G_k$ with residue $e^{j(\phi_k-\theta_0)}r$.
Thus, for $\phi_k=\theta_0$, condition c) of Lem.~\ref{lem:positive function} reduces to $r\succeq0$.
Conversely, if $G_k(s)$ is biproper and has no poles on $\partial\calD_0$, condition b) of Lem.~\ref{lem:positive function} restricts $\phi_k$ to a closed interval. The owner computes this interval from its own model and either uses $\phi_k=\theta_0$ if admissible or reports one admissible $\phi_k$ to the operator without revealing its model.
(iii) \textit{Network condition.} The condition $\hat Y+\hat Y^\HT\succeq0$ involves no subsystem model, only $Y$ and $\Phi$, and the operator knows $\Phi$ by (ii). Hence, Th.~\ref{thm: D stable with positivity} has the form required in Problem~\ref{prob:main}. The operator evaluates the network condition, while each owner verifies the positivity of its subsystem. With $\Phi=\theta_0 I$,
\begin{equation}
\hat Y+\hat Y^\HT=\cos\theta_0\,(Y+Y^\top)-j\sin\theta_0\,(Y-Y^\top).
\end{equation}
For a symmetric $Y$, the condition reduces to $\cos\theta_0\, Y\succeq0$, i.e., $Y\succeq0$ if $\theta_0<\pi/2$, and holds trivially if $\theta_0=\pi/2$. In particular, if the subsystems are coupled only through weighted output differences, $w_k=\sum_{l}W_{kl}(\eta_l-\eta_k)$ with weights $W_{kl}=W_{lk}\succeq0$, then $Y$ is a weighted Laplacian, which is positive semi-definite, and the network condition holds for every $\theta_0\in[0,\pi/2]$. Resistive electrical networks, such as the DCMG in Section~\ref{section: Application to DCMG}, are of this type. For a nonsymmetric $Y$ (directed coupling), the skew-symmetric part enters with the weight $\sin\theta_0$, and the condition may fail. It can be restored by the loop transformation below.
\end{rem}}


\begin{figure}[t]
    \centering
\includegraphics[width=0.5\linewidth]{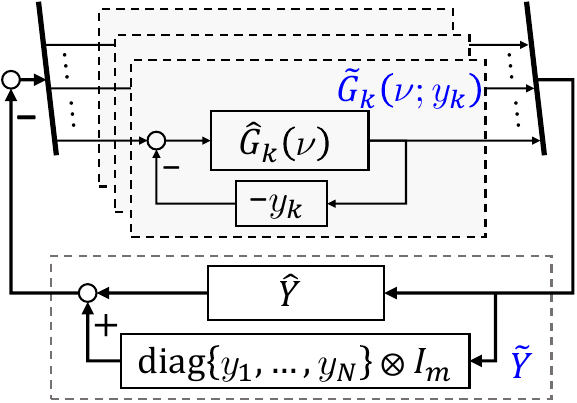}
    \caption{Loop transformation.}
    \label{fg:loop transformation}
\end{figure}

Positivity, like PR, can be shaped between the subsystems and the network by loop transformations \cite{brogliato2020dissipative}. As shown in Fig.~\ref{fg:loop transformation}, a local feedback with gain $y_k\in\mathbb{R}$ around each rotated subsystem $\hat G_k(\nu)$ is cancelled by a feedforward on the network $\hat Y$. The closed loop is unchanged and equals the feedback interconnection of the modified network ${\tilde{Y} = \hat{Y} + \diag{y_1, \dots, y_N}\otimes I_m}$ and the modified subsystems
\begin{equation}\label{eq:modified subsystem}
\tilde{G}_k(\nu;y_k) = [I - y_k\hat{G}_k(\nu)]^{-1}\hat{G}_k(\nu).
\end{equation}
Hence, $\calD$-stability is guaranteed if each subsystem satisfies its \textit{local condition}, i.e., $\tilde G_k(\nu;y_k)$ is a positive function, and $\tilde Y$ satisfies the \textit{network condition} $\tilde Y+\tilde Y^\HT\succeq0$. The latter still involves only $Y$, $\Phi$, and the indices $y_k$. We call $y_k$ the \textit{positivity index}, since it plays the role of the output-feedback passivity index \cite{he2024passivity,peng2026positive}. A smaller index relaxes the local condition and tightens the network condition. Precisely, positivity is monotonic in the index. If $\tilde G_k(\nu;y_{k,1})$ is positive, then $\tilde G_k(\nu;y_{k,2})$ is positive for all $y_{k,2}<y_{k,1}$, because it is the negative feedback interconnection of $\tilde G_k(\nu;y_{k,1})$ with the gain $(y_{k,1}-y_{k,2})I\succ0$.

\begin{rem}\label{rem:KYP}
By Lem.~\ref{lem:positive function} and the KYP lemma \cite{brogliato2020dissipative}, $\tilde G_k(\nu;y_k)$ is positive if and only if an LMI for a minimal realization of the real-rational $\la\tilde G_k\ra(\nu;y_k)$ has a solution $P\succ0$. \rev{This LMI is solved when the controller is designed or re-checked, not in real time.} For SISO subsystems, positivity of $\tilde G_k(\nu;y_k)$ can be checked graphically. 
Its condition b) holds if and only if the plot of $\re{1/\hat G_k(j\omega)}$ does not fall below $y_k$ where $\hat G_k(j\omega)\neq0$ (cf. \cite[Th.~2]{peng2026positive}).
Conditions a) and c) are checked by the Nyquist criterion for ${1-y_k\hat G_k(\nu)}$ over $\omega \in (-\infty, \infty)$, since the plot of a complex-coefficient $\hat G_k$ is not symmetric (cf. \cite[Th.~3]{peng2026positive}).
\end{rem}

\section{Application to DC Microgrids}\label{section: Application to DCMG}

\subsection{Small-Signal Modeling and Subsystem Partitioning}\label{subsection:DC Microgrid Modeling}
The DCMG comprises $N$ nodes interconnected via a transmission network. Each node $k\in\mathscr{V}=\{1,\dots,N\}$ hosts an ESS, a PV, or a CPL, which injects the current $i_k$ into the network at the node voltage $u_k$. The transmission network is modeled as a linear resistive grid via the admittance matrix $Y=Y^\top\succeq0$ \cite{sun2026passivity}:
\begin{equation}\label{eq:netwrok Y}
\bm i=Y \bm u
\end{equation}
Around an equilibrium $(\bm i^*,\bm u^*)$, the small-signal transfer function $G_k(s)$ maps $-\Delta i_k$ to $\Delta u_k$, so that \eqref{eq:subsystem}--\eqref{eq:network} hold with $w_k=-\Delta i_k$, $\eta_k=\Delta u_k$, and $Y$ from \eqref{eq:netwrok Y}.

We adopt the device models of \cite{sun2026passivity}. An ESS is interfaced by a boost or a buck converter (boost or buck ESS) and regulates its port voltage to the droop reference $U_k^\mr-R_k^\md i_k$ by cascaded voltage-current PI loops. With the fast current loop neglected by time-scale separation \cite{khalil2002nonlinear}, its transfer function has the second-order form $G_k(s) = (c_1 s + c_0)/(s^2 + d_1 s + d_0)$, where $c_1,c_0>0$ and $d_1, d_0 \in \mathbb{R}$ depend on the hardware parameters, controller gains, and the equilibrium voltage $u_k^*$ (Table~\ref{tab:parameter_mapping}) \cite{sun2026passivity}. \rev{A PV or a CPL delivers or draws a constant power $P_k$ at its port capacitor $C_k$, and its transfer function is \cite{sun2026passivity}
\begin{equation}\label{eq:fixed device}
G_k(s)=\frac{1}{C_k s+g_k},\quad g_k=\pm P_k/(u_k^*)^2,
\end{equation}
with $+$ for a PV and $-$ for a CPL. For a CPL, $g_k<0$ is the well-known negative incremental conductance \cite{laib2023decentralized}. Since \eqref{eq:fixed device} contains no controller gains, the port behavior of a PV or a CPL cannot be shaped by its controller. We thus call PVs and CPLs \textit{fixed devices}, following \cite{sun2026passivity}, and collect them in $\mathscr V^\mrm f$. The ESSs (referred to as \textit{tunable devices}) are collected in $\mathscr V^\mrm t$. $Y$ is partitioned accordingly into the blocks $Y^\mrm{tt}$, $Y^\mrm{tf}$, $Y^\mrm{ft}$, and $Y^\mrm{ff}$, and the indices $y_k$ into the vectors $\mathbf y^\mrm t$ and $\mathbf y^\mrm f$.}

\begin{table}[t]
\caption{Coefficients of the ESS Transfer Functions}
\label{tab:parameter_mapping}
\footnotesize
\begin{tabular*}{\columnwidth}{@{\extracolsep{\fill}}lcccc@{}}
\toprule
Device & $c_1$ & $c_0$ & $d_1$ & $d_0$ \\
\midrule
Boost ESS & $\frac{E k_\mP^u R^\md + u^*}{C u^*}$ & $\frac{E R^\md k_\mI^u}{C u^*}$ & $\frac{U^\mr - u^* + E R^\md k_\mP^u}{C u^* R^\md}$ & $\frac{E k_\mI^u}{C u^*}$ \\
\addlinespace[4pt]
Buck ESS & $\frac{R^\md k_\mP^u + 1}{C}$ & $\frac{R^\md k_\mI^u}{C}$ & $\frac{k_\mP^u}{C}$ & $\frac{k_\mI^u}{C}$ \\
\bottomrule
\end{tabular*}
\par\smallskip
Note: $C$, $(k_\mP^u,k_\mI^u)$, $R^\md$, $E$, $U^\mr$, and $u^*$ denote the port capacitance, voltage-loop PI gains, droop coefficient, battery voltage, voltage reference, and equilibrium voltage of ESS $k$, respectively. The subscript $k$ is omitted.
\end{table}

\subsection{Decentralized Certification and Nominal Grid Code}\label{subsection: loop transformation}
\rev{We consider target regions with $\theta_0\in[0,\frac{\pi}{2}]$, $\omega_0\ge0$, and $\sigma_0\le0$, such as $\calD_\mrm{LHP}$, $\calD_\mrm{SEC}$, and $\calD_\mrm{HS}$. With $\Phi=\theta_0 I$ (Rem.~\ref{rem:theorem conditions}), a fixed device \eqref{eq:fixed device} has the rotated transfer function $\hat G_k(\nu)=(C_k\nu+\bar y_k+jq_k)^{-1}$ with $q_k\triangleq C_k\omega_0\cos\theta_0-g_k\sin\theta_0$ and
\begin{equation}\label{eq:fixed index}
\bar y_k\triangleq g_k\cos\theta_0+C_k(\sigma_0+\omega_0\sin\theta_0).
\end{equation}
If $\bar y_k<0$, which is typical for a CPL due to $g_k<0$, the pole $-(\bar y_k+jq_k)/C_k$ lies in $\re{\nu}>0$. Then $\hat G_k(\nu)$ is not positive, no tuning can change this, and Th.~\ref{thm: D stable with positivity} does not apply directly. With a suitable positivity index, however, the loop transformation renders fixed devices positive.

\begin{lem}\label{lem:fixed device}
For a fixed device, $\tilde G_k(\nu;y_k)$ in \eqref{eq:modified subsystem} is a positive function if and only if $y_k\le\bar y_k$.
\end{lem}
\textit{Proof:} We have $\tilde G_k(\nu;y_k)^{-1}=\hat G_k(\nu)^{-1}-y_k=C_k\nu+\bar y_k-y_k+jq_k$. Hence, $\tilde G_k(\nu;y_k)$ has the single pole $-(\bar y_k-y_k+jq_k)/C_k$ with the residue $1/C_k>0$, and $\re{\tilde G_k(j\omega;y_k)}=(\bar y_k-y_k)\,|\tilde G_k(j\omega;y_k)|^2$. Thus, conditions a)--c) of Lem.~\ref{lem:positive function} hold if and only if $y_k\le\bar y_k$.\hfill$\square$}

\rev{Hence, $\bar y_k$ is the largest index of a fixed device, and assigning $y_k=\bar y_k$ satisfies its local condition without any action of its owner while maximally relaxing the network condition. In effect, this loop transformation incorporates $\bar y_k$ directly into the modified network matrix $\tilde Y$. This assignment is implementable by the operator for two reasons: (i) it is a purely mathematical construct and adds no physical hardware or control loop; (ii) the operator can compute $\bar y_k$ by \eqref{eq:fixed index} without any controller model, from the port capacitance $C_k$, specified in the interconnection agreement, and from $P_k$ and $u_k^*$, obtained from forecasts or prescribed ranges and a power-flow computation (Section~\ref{subsection: robust grid code}).}

The ESSs, in contrast, meet their local conditions by tuning. Since the network condition $\tilde Y+\tilde Y^\HT\succeq0$ is a coupled matrix inequality in the indices $\mathbf y^\mrm t$ of all ESSs, we replace it by a common lower bound that each ESS can meet on its own. \rev{Let $\Psi\triangleq\cos\theta_0\,Y^\mrm{ff}+\diag{\mathbf y^\mrm f}$ and define the Schur complement}
\begin{equation}\label{eq:Xi_matrix}
\rev{\Xi\triangleq\cos\theta_0\,Y^\mrm{tt}-\cos^2\theta_0\,Y^\mrm{tf}\Psi^{-1}Y^\mrm{ft}.}
\end{equation}
\begin{thm}\label{thm: uniform bound}
\rev{Let $y_k\le\bar y_k$ for $k\in\mathscr V^\mrm f$, assuming $\Psi\succ0$.} The DCMG is $\mathcal{D}$-stable w.r.t. $\mathcal{D}(\theta_0, \omega_0, \sigma_0)$ if \rev{each ESS $k\in\mathscr{V}^\mrm{t}$ admits an index $y_k$ such that}:
\begin{enumerate}[label={\roman*)}]
\item\label{con: cor net} Network condition: $y_k\ge-\lambda_\mrm{min}(\Xi)$.
\item\label{con: cor device} Local condition: $\tilde{G}_k(\nu;y_k)$ is a positive function.
\end{enumerate}
\end{thm}

\textit{Proof:}
By Lem.~\ref{lem:fixed device} and ii), all $\tilde G_k(\nu;y_k)$, $k\in\mathscr V$, are positive. Since $Y$ is real symmetric, $\tilde{Y} + \tilde Y^\HT = 2(\cos{\theta_0} Y+\diag{\mbf{y}^\mrm{t},\mathbf{y}^\mrm f})$. As $\Psi\succ0$, the Schur complement lemma shows that $\tilde{Y} + \tilde Y^\HT \succeq 0$ if and only if $\diag{\mathbf{y}^\mrm t} + \Xi \succeq 0$. This is implied by i), since $\diag{\mathbf{y}^\mrm t} + \Xi \succeq (\min_{k\in\mathscr V^\mrm t} y_k+\lambda_\mrm{min}(\Xi)) I$. Th.~\ref{thm: D stable with positivity} completes the proof.\hfill$\square$

\rev{Thus, the system operator specifies a grid code, which consists of the target region $\calD(\theta_0,\omega_0,\sigma_0)$ (encoding the required dynamic performance), and the required index $-\lambda_\mrm{min}(\Xi)$. Each ESS verifies ii) using its own model only.}

\begin{rem} \textit{(Physical limit and operational intervention)}\label{remark:Feasibility of Admittance Shifting}
\rev{For $\mathbf y^\mrm f=\bar{\mathbf y}^\mrm f$, i.e., $y_k=\bar y_k$ for all fixed devices, the assumption $\Psi\succ0$ holds if and only if the DCMG with all ESSs replaced by ideal voltage sources at their equilibrium voltages has all its modes in the interior of $\calD$.
It may fail for heavily loaded CPLs that are weakly connected to the ESSs or for a large required decay rate $|\sigma_0|$. If $\Psi$ has a negative eigenvalue even at $\mathbf y^\mrm f=\bar{\mathbf y}^\mrm f$, no tuning of the ESSs can satisfy the network condition. In this case, or if an ESS cannot attain the required index,} the operator must intervene by relaxing the target region\rev{, reinforcing the network,} or shedding non-critical CPLs.
\end{rem}

\rev{\subsection{Robust Certification over an Operating Set}\label{subsection: robust grid code}
We certify $\calD$-stability for all operating points in a prescribed set, so that the required index and the ESS controllers need not be updated when the PV and CPL powers vary within this set. This set is $\Omega\triangleq\{\bm P: P_k^-\le P_k\le P_k^+,\ k\in\mathscr V^\mrm f\}$, where $\bm P$ collects the powers $P_k$ of PVs and CPLs. Suppose that the operating point with minimum generation and maximum demand admits a long-term voltage-stable power-flow solution $\bm u^-$ \cite[Def.~2.8]{jeeninga2022dcPartI}. This holds if and only if this operating point lies strictly within the loadability limit \cite[Cor.~3.21]{jeeninga2022dcPartI}. Such a solution is the desired one, with the highest voltages among all solutions \cite[Th.~8.2]{jeeninga2022dcPartII}. Under this assumption, every $\bm P\in\Omega$ admits a unique long-term voltage-stable power-flow solution $\bm u^*(\bm P)$, which is nondecreasing in PV powers and nonincreasing in CPL powers \cite{jeeninga2022dcPartI,jeeninga2022dcPartII}, \cite[App.~A]{sun2026passivity}. Hence, with $\bm u^+$ denoting $\bm u^*(\bm P)$ at maximum generation and minimum demand,
\begin{equation}\label{eq:voltage bounds}
\bm u^-\le\bm u^*(\bm P)\le\bm u^+,\quad\forall\bm P\in\Omega.
\end{equation}
Define the worst-case conductances $g_k^-\triangleq P_k^-/(u_k^+)^2$ for PVs and $g_k^-\triangleq -P_k^+/(u_k^-)^2$ for CPLs, and let $\bar y_k^-$ be the resulting lower bounds of $\bar y_k$ by \eqref{eq:fixed index}, collected in $\bar{\mathbf y}^{\mrm f,-}$. Define $\Psi^-$ and $\Xi^-$ as $\Psi$ and \eqref{eq:Xi_matrix} evaluated at $\mathbf y^\mrm f=\bar{\mathbf y}^{\mrm f,-}$.

\begin{cor}\label{cor:robust grid code}
Suppose that $\theta_0\in[0,\frac{\pi}{2}]$, \eqref{eq:voltage bounds} holds, and $\Psi^-\succ0$. If every ESS $k\in\mathscr V^\mrm t$ admits an index $y_k\ge-\lambda_\mrm{min}(\Xi^-)$ such that $\tilde G_k(\nu;y_k)$ is a positive function for every $u_k^*\in[u_k^-,u_k^+]$, then the DCMG is $\calD$-stable w.r.t. $\calD(\theta_0,\omega_0,\sigma_0)$ at $\bm u^*(\bm P)$ for every $\bm P\in\Omega$.
\end{cor}
\textit{Proof:} Consider any $\bm P\in\Omega$. By \eqref{eq:voltage bounds}, each ESS satisfies ii) of Th.~\ref{thm: uniform bound}, and each fixed device has $g_k\ge g_k^-$, since $P_k\in[P_k^-,P_k^+]$. As $\cos\theta_0\ge0$, \eqref{eq:fixed index} yields $\bar y_k\ge\bar y_k^-$, and Th.~\ref{thm: uniform bound} applies with $y_k=\bar y_k^-$, $k\in\mathscr V^\mrm f$.\hfill$\square$

\begin{rem}\label{rem:topology}
\textit{(Topology and line-parameter changes)}
The required index $-\lambda_\mrm{min}(\Xi^-)$ and the voltage intervals depend on $Y$. A required index that remains valid for a set of credible topologies (e.g., all $N-1$ line outages) or line-parameter variations is the maximum of $-\lambda_\mrm{min}(\Xi^-)$ over this set, with voltage intervals that contain those of all its members. Such changes then require no update.
\end{rem}}

\subsection{Decentralized Parameter Synthesis and Implementation}\label{subsection: Decentralized Parameter Synthesis}
Th.~\ref{thm: uniform bound} and Cor.~\ref{cor:robust grid code} enable decentralized, non-iterative parameter synthesis. For second-order subsystems, the following algebraic test makes it explicit.

\begin{prop}\label{prop:second-order coefficient condition}
A complex-coefficient second-order transfer function $h(\nu) = (a_1 \nu + a_0)/(\nu^2 + b_1 \nu + b_0)$, where $a_l = a_{lr} + j a_{li}$ and $b_l = b_{lr} + j b_{li}$, $l=0,1$, is a positive function if:
\begin{enumerate}[label=\roman*)]
\item Strict stability: $b_{1r} \!> \!0$ and $b_{1r}^2 b_{0r} + b_{1r}b_{1i}b_{0i} - b_{0i}^2 \!> \!0$.
\item Real-part non-negativity: $a_{1i} = 0$, $p_2\ge0$, $p_0\ge0$, and $p_1^2\le4p_2p_0$, where $p_2\triangleq a_{1r}b_{1r}-a_{0r}$, $p_1\triangleq a_{0i}b_{1r}+a_{1r}b_{0i}-a_{0r}b_{1i}$, and $p_0\triangleq a_{0r}b_{0r}+a_{0i}b_{0i}$.
\end{enumerate}
\end{prop}
\textit{Proof:} Condition~i) states that the Hurwitz determinants of $D(\nu)\triangleq\nu^2+b_1\nu+b_0$ are positive, so both roots of $D$ lie in $\re{\nu}<0$ \cite[Th.~(40,1)]{marden1949geometry}, and conditions a) and c) of Lem.~\ref{lem:positive function} hold. Moreover, $\re{h(j\omega)}=N(\omega)/|D(j\omega)|^2$ with $N(\omega)=a_{1i}\,\omega(\omega^2+b_{1i}\omega-b_{0r})+p_2\omega^2+p_1\omega+p_0$.
Under ii), $N(\omega)\ge0$ for all $\omega$, which is b).\hfill$\square$

For an ESS with model $G_k(s)$ in Section~\ref{subsection:DC Microgrid Modeling}, Prop.~\ref{prop:second-order coefficient condition} yields explicit constraints that guarantee the positivity of $\tilde G_k(\nu;y_k)$ in the three regions in \eqref{eq:performance regions}. Since a smaller index relaxes the local condition (Section~\ref{subsec:main result}), the ESS only needs to check at $y_k$ equal to the required index, the smallest admissible value. For $\calD_\mrm{LHP}(\alpha)$, $s(\nu)=\nu+\alpha$ yields
\begin{equation}
\tilde{G}_k(\nu;y_k)=\frac{c_1 \nu + c_1\alpha + c_0}{\nu^2+b_1\nu + b_0},
\end{equation}
with $b_1=d_1+2\alpha-c_1y_k$ and $b_0=\alpha^2+ d_1\alpha+d_0-(c_1\alpha+ c_0)y_k$. The resulting constraints are $c_1\alpha+c_0>0$, $y_k \le (c_1(d_1+\alpha)-c_0)/c_1^2$, and $y_k < (\alpha^2+d_1\alpha+d_0)/(c_1\alpha+c_0)$. For $\calD_\mrm{SEC}(\beta)$, $s(\nu) = e^{j(\pi/2 - \beta)}\nu$ yields the constraints $y_k < d_0 \sin\beta/c_0$ and $y_k < (c_1 d_1 - c_0)/(c_1^2 \sin\beta) - d_0 \cos^2\beta/(c_0 \sin\beta)$. For $\calD_\mrm{HS}(\gamma)$, $\cos\theta_0=0$ gives $\Xi=\Xi^-=0$ and thus the required index~$0$. Then $s(\nu) = j\nu+j\gamma$ and $y_k=0$ yield the constraint $\bar\gamma_k< \gamma$, where $\bar\gamma_k \triangleq \sqrt{\max\{c_0^2 - c_0c_1d_1 + c_1^2d_0, 0\}}/c_1$.

\rev{In summary, the operator obtains $\bm u^-$ and $\bm u^+$ from two power-flow computations, the bounds $\bar y_k^-$ of fixed devices by \eqref{eq:fixed index}, and the required index $-\lambda_\mrm{min}(\Xi^-)$ for each sub-region $\calD_i$ of $\calDcap$ (Rem.~\ref{rem:intersection}). 
These computations need no controller model, since the power flow depends only on steady-state parameters fixed in a prior dispatch stage (e.g., droop coefficients and voltage references).
The operator provides the grid code once, offline or by a broadcast: the sub-regions, their required indices, and the voltage intervals $[u_k^-,u_k^+]$ (or a common voltage band containing them). 
Using only its own model, each ESS then designs its controller to satisfy its local condition of Th.~\ref{thm: uniform bound} with $y_k=-\lambda_\mrm{min}(\Xi^-)$ for each sub-region and each $u_k^*\in[u_k^-,u_k^+]$. For the second-order ESSs, the above constraints ensure it, and the owner can choose any gains that satisfy them. 
Since the local condition concerns only $G_k(s)$, the owner may even change the controller structure.
Dynamic parameters (e.g., PI gains) do not affect the steady-state power flow and can be retuned locally without involving the operator or other ESSs.}

\rev{After the network or the set of devices changes, the operator recomputes the grid code and broadcasts it once. An ESS re-checks its local condition only if its required index increases or its new voltage interval is not contained in the previous one. If the re-check fails, the ESS retunes its controller with the same local test. Hence, this test (Rem.~\ref{rem:KYP}) is solved only at commissioning and in such re-checks, not in the real-time control loop.}

\section{Case Study}\label{section: case study}
\subsection{Setup and Default Design}\label{subsec:case setup}
The proposed method is validated on a DCMG adopting the IEEE 39-node topology \cite{canizares2015benchmark}. Each line resistance is $0.1~\Omega$. The DCMG comprises boost ESSs (Nodes 1, 2, 5, 6, 9, 10, 11), buck ESSs (Nodes 13, 14, 16, 17, 19, 22, 28), PVs (Nodes 30--39), and CPLs (Nodes 3, 4, 7, 8, 12, 15, 18, 20, 21, 23, 24, 25, 26, 27, 29). Each PV consists of 7 parallel strings of the QJM200-72 module (provided in MATLAB/Simulink) and delivers the nominal power $P_0^\mpv=1400.1~\mrm{W}$ at its maximum power point ($36.12~\mrm{V}$) under $1000~\mrm{W/m^2}$ and $25^\circ$C. Each CPL consumes the nominal power $P_0^\ml=1500~\mrm{W}$. The physical and control parameters are summarized in Table~\ref{tab:system_parameters}. The time-domain results are obtained from nonlinear averaged models of the converters. The target pole region is set as $\calDcap= \mathcal{D}_\mrm{LHP}(\alpha) \cap \mathcal{D}_\mrm{SEC}(\beta) \cap \mathcal{D}_\mrm{HS}(\gamma)$ with $\alpha=-11$, $\beta=5\pi/12$, and $\gamma=24\pi$. It requires a decay rate of at least $11~\mrm{s^{-1}}$ (settling time of about $0.36~\mrm{s}$), a damping ratio of at least $\cos\beta=0.26$, and frequencies below $75.4~\mrm{rad/s}$ ($12~\mrm{Hz}$). \rev{In the operating set $\Omega$, the power of each PV ranges over $[0.9,1]\,P_0^\mpv$ (irradiance variations) and that of each CPL over $[0.85,1.15]\,P_0^\ml$. For numerical checks, we sample $\Omega$ at $35$ operating points, at which all PVs deliver $aP_0^\mpv$ and all CPLs consume $bP_0^\ml$, with $5$ uniformly spaced values $a\in[0.9,1]$ and $7$ uniformly spaced values $b\in[0.85,1.15]$. The nominal point is $(a,b)=(1,1)$.}

\begin{table}[t]
\caption{Physical and Control Parameters of the DCMG}
\label{tab:system_parameters}
\footnotesize
\setlength{\tabcolsep}{2pt}
\renewcommand{\arraystretch}{1.3}
\begin{tabular*}{\columnwidth}{@{\extracolsep{\fill}}lcccc@{}}
\toprule
Parameter & Boost ESS & Buck ESS & PV & CPL \\
\midrule
Inductor (mH) & 0.5 & 0.5 & 0.5 & 0.5 \\
Capacitor $C$ (mF) & 4 & 4 & 2 & 2 \\
Battery $E$ (V) & 50 & 200 & -- & -- \\
Voltage ref.\ (V) & 105 & 105 & 36.12 & 50 \\
Droop $R^\md$ ($\Omega$) & 0.6 & 0.7 & -- & -- \\
Current PI & 0.02,\,40 & 0.02,\,40 & 0.02,\,20 & 0.01,\,10 \\
Default voltage PI & 0.01,\,100 & 0.01,\,70 & 1,\,50 & 0.5,\,50 \\
Tuned voltage PI & 1.2,\,52 & 1.2,\,29 & -- & -- \\
\bottomrule
\end{tabular*}
\par\smallskip
The PV array-side and CPL load-side capacitors are $4~\mrm{mF}$.
\end{table}

With the default voltage-loop gains, the DCMG is stable but not $\calDcap$-stable. At the nominal point, its dominant pole pair $-8.46\pm j86.0$ violates all three requirements (decay rate $8.46<11~\mrm{s^{-1}}$, damping ratio $0.10<0.26$, frequency $86.0>75.4~\mrm{rad/s}$). The same holds at all sampled operating points (Fig.~\ref{fg:poles}(a)), and the responses are oscillatory (Section~\ref{subsec:case time}).
\rev{\begin{figure}[t]
    \centering\revcolor
\includegraphics[width=0.9\linewidth]{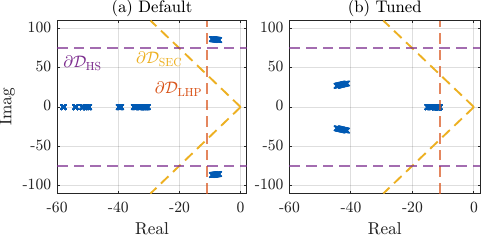}
    \caption{Closed-loop poles at the $35$ sampled operating points with (a) default and (b) tuned gains.}
    \label{fg:poles}
\end{figure}

\begin{figure}[t]
    \centering\revcolor
\includegraphics[width=0.9\linewidth]{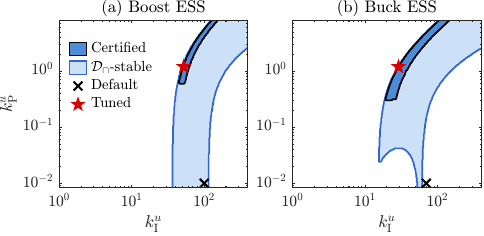}
    \caption{Gain planes of (a) the boost and (b) the buck ESSs. Certified: by Cor.~\ref{cor:robust grid code}. $\calDcap$-stable: closed-loop poles in $\calDcap$ at all sampled operating points, with the buck ESSs in (a) and the boost ESSs in (b) at their tuned gains.}
    \label{fg:design map}
\end{figure}

\subsection{Decentralized Synthesis}\label{subsec:case synthesis}
\textit{Grid code.} The operator computes $\bm u^-$ at minimum generation and maximum demand, $(a,b)=(0.9,1.15)$, and $\bm u^+$ at maximum generation and minimum demand, $(a,b)=(1,0.85)$. Since $\bm u^-$ is long-term voltage stable (checked by \cite[Prop.~3.3]{jeeninga2022dcPartI}), \eqref{eq:voltage bounds} holds, which all sampled operating points confirm. The assumption $\Psi^-\succ0$ of Cor.~\ref{cor:robust grid code} holds for all three sub-regions. The operator thus provides the robust grid code. It consists of the target region $\calDcap$ (i.e., $\alpha$, $\beta$, and $\gamma$), the required indices $0.1617$, $0.1168$, and $0$ for $\calD_\mrm{LHP}$, $\calD_\mrm{SEC}$, and $\calD_\mrm{HS}$, respectively, and the voltage band $[95.59,103.63]~\mrm{V}$, which contains the voltage intervals of all ESSs.

\textit{Local tuning.} Each owner tunes the voltage-loop PI gains $(k_\mP^u,k_\mI^u)$ of its ESS, which do not affect the steady-state power flow. Fig.~\ref{fg:design map} shows the gains certified by Cor.~\ref{cor:robust grid code} for all ESSs of each type. Using only its own model, each owner sets, e.g., $k_\mP^u=1.2$, and chooses the $k_\mI^u$ that maximizes the smallest relative margin of the bounds on $y_k$ and $\bar\gamma_k$ in Section~\ref{subsection: Decentralized Parameter Synthesis} over the voltage band. 
Here, the relative margin is $(y_k+\lambda_\mrm{min}(\Xi^-))/|\lambda_\mrm{min}(\Xi^-)|$ for $\calD_\mrm{LHP}$ and $\calD_\mrm{SEC}$, with $y_k$ the largest index allowed by the constraints, and $(\gamma-\bar\gamma_k)/\gamma$ for $\calD_\mrm{HS}$.
Since the ESSs of each type have the same hardware and use the same voltage band, their owners independently obtain the same $k_\mI^u$: $52$ for every boost ESS and $29$ for every buck ESS. Over the voltage band, and hence over its own voltage interval, every ESS satisfies its local conditions. The largest indices allowed by the constraints are at least $0.380$ (boost) and $0.604$ (buck) for $\calD_\mrm{LHP}$ and $0.312$ and $0.551$ for $\calD_\mrm{SEC}$, above the required indices $0.1617$ and $0.1168$; $\bar\gamma_k$ is at most $68.7$ and $63.7~\mrm{rad/s}$, below $\gamma=75.4~\mrm{rad/s}$. Each owner uses only its own model and the grid code. Neither the network nor the models of other devices are needed.

\textit{Verification.} By Cor.~\ref{cor:robust grid code}, the tuned DCMG is $\calDcap$-stable for every $\bm P\in\Omega$, which Fig.~\ref{fg:poles}(b) confirms at the sampled operating points. At the nominal point, the least-damped poles are $-43.3\pm j28.2$ (damping ratio $0.84$). At all sampled points, the frequencies stay below $31~\mrm{rad/s}$, and the slowest pole is $-11.23$, near $\alpha=-11$.


\begin{table}[t]
\centering
\caption{Worst-Case Nominal Responses (Default $\to$ Tuned)}
\label{tab:metrics}
\footnotesize
\setlength{\tabcolsep}{3pt}
\begin{tabular}{@{}lcc@{}}
\toprule
Metric & $1\%$ pulse & \rev{$+10\%$ step} \\
\midrule
Max.\ voltage deviation (V) & $0.201\to0.155$ & \rev{$2.72\to2.57$} \\
Voltage overshoot (V) & $0.148\to0.034$ & \rev{$1.15\to0.93$} \\
Voltage settling time (s) & $0.48\to0.16$ & \rev{$0.49\to0.16$} \\
Max.\ ESS current deviation (A) & $0.123\to0.111$ & \rev{$2.45\to2.21$} \\
ESS current overshoot (A) & $0.056\to0.009$ & \rev{$0.45\to0.25$} \\
ESS current settling time (s) & $0.45\to0.15$ & \rev{$0.35\to0.17$} \\
\bottomrule
\multicolumn{3}{@{}l}{\footnotesize Overshoot beyond the final value; $2\%$ settling band.}
\end{tabular}
\end{table}

\begin{figure}[t]
    \centering\revcolor
\includegraphics[width=\linewidth]{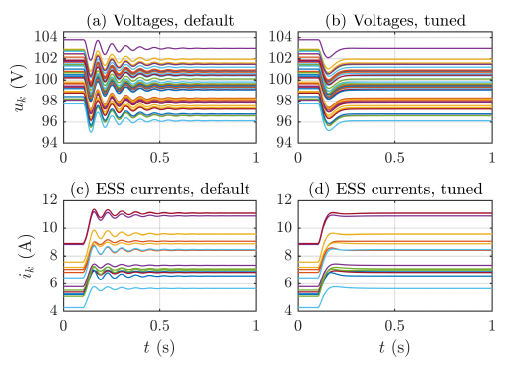}
    \caption{Responses to a permanent $10\%$ increase of all CPL powers at $t=0.1~\mrm{s}$: (a), (b) node voltages and (c), (d) injected ESS currents with default and tuned gains.}
    \label{fg:step}
\end{figure}

\begin{figure}[t]
    \centering\revcolor
\includegraphics[width=0.98\linewidth]{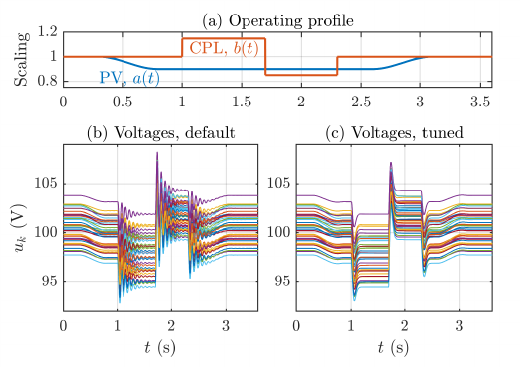}
    \caption{Responses along an operating profile within $\Omega$ with fixed gains and grid code: (a) PV and CPL power scaling factors; node voltages with (b) default and (c) tuned gains.}
    \label{fg:profile}
\end{figure}}

\subsection{Time-Domain Validation}\label{subsec:case time}
Table~\ref{tab:metrics} summarizes the responses at the nominal point to a pulse load disturbance (all CPL powers increased by $1\%$ for $20~\mrm{ms}$) and to a permanent $10\%$ load step, both at $t=0.1~\mrm{s}$. Each metric is the worst case over all nodes (voltages) or all ESSs (currents). With the tuned gains, every metric improves.

\rev{\textit{Permanent load step.} Fig.~\ref{fg:step} shows the node voltages and the injected ESS currents after the permanent step. Since the PI gains do not affect the steady-state power flow, the synthesis thus improves only the damping of the transients and leaves the steady-state voltages and the power sharing of droop control unchanged.

\textit{Operating-point variations.} To illustrate Cor.~\ref{cor:robust grid code}, Fig.~\ref{fg:profile} shows the responses along an operating profile within $\Omega$, with fixed gains and without recomputing the grid code. The PV generation decreases smoothly to $90\%$ and recovers, and the CPL powers step to $115\%$, $85\%$, and $100\%$ at $t=1.0$, $1.7$, and $2.3~\mrm{s}$. After the three load steps, the tuned design settles within $0.16$, $0.15$, and $0.15~\mrm{s}$ (default: $0.60$, $0.42$, and $0.29~\mrm{s}$), with overshoots of $1.42$, $3.09$, and $1.42~\mrm{V}$ (default: $1.71$, $4.12$, and $1.93~\mrm{V}$).

\subsection{Robustness, Conservatism, and Comparison}\label{subsec:case robustness}
\textit{Line parameters and topology.} For a perturbed network, the operator recomputes $-\lambda_\mrm{min}(\Xi^-)$ and the voltage intervals. The certificate is \textit{retained} if $\Psi^-\succ0$ holds and each ESS still satisfies its local condition with its tuned gains. 
(i) Under a common scaling of all line resistances, it is retained for factors from $0.25$ up to $6.03$, where it is lost. 
(ii) Under independent, uniformly distributed perturbations of every line resistance within $\pm50\%$, it is retained for all $200$ samples, and $-\lambda_\mrm{min}(\Xi^-)$ varies only within $[0.158,0.164]$ ($\calD_\mrm{LHP}$) and $[0.114,0.119]$ ($\calD_\mrm{SEC}$). 
(iii) It is also retained for all $35$ single-line outages that keep the network connected, even with a single grid code (Rem.~\ref{rem:topology}): its required indices are $0.173$ and $0.126$, and its voltage band is $[91.76,105.44]~\mrm{V}$.
Over this band, the largest indices of boost ESSs become at least $0.353$ and $0.284$, and $\bar\gamma_k$ at most $68.8~\mrm{rad/s}$. Those of the buck ESSs are unchanged, since their model does not depend on $u^*$ (Table~\ref{tab:parameter_mapping}). Hence, the tuned ESSs satisfy this grid code, and no single-line outage requires an update.

\textit{Load power.} We keep the tuned gains and increase all CPL powers uniformly, with the PV generation at its minimum of $90\%$. The operator recomputes the required indices for each load level. The certificate is retained up to $2.21$ times the nominal CPL power, close to the actual limit of $\calDcap$-stability at $2.27$ times. The equilibrium disappears at the loadability limit of $2.51$ times.

\textit{Conservatism.} Since Th.~\ref{thm: D stable with positivity} is sufficient, no certified gain violates $\calDcap$. For comparison, Fig.~\ref{fg:design map} shows the $\calDcap$-stable gains, which are computed from the models of all devices and depend on the gains of the other ESS type. In contrast, the certified gains of each ESS are determined by its own model and the grid code alone. 
In Fig.~\ref{fg:design map}, the certified set covers $5.7\%$ (boost) and $11.7\%$ (buck) of the $\calDcap$-stable set. To isolate the conservatism of the certificate itself, we repeat the comparison at the nominal point with Th.~\ref{thm: uniform bound}. The coverages are then $6.7\%$ and $10.6\%$, so that covering all of $\Omega$ adds little conservatism. If each sub-region is considered alone, i.e., both the certificate and the closed-loop poles are checked only for $\calD_\mrm{LHP}$, $\calD_\mrm{SEC}$, or $\calD_\mrm{HS}$, the coverages are at least $32\%$, $42\%$, and $74\%$ for both ESS types. 
The conservatism is thus compounded by intersecting the sub-regions, whereas in each sub-region it stems from the separate positivity of the devices and the network with uniform indices.


\textit{Comparison.} (i) A centralized LMI test \cite{chilali1996h} is exact for analysis at an operating point and confirms the tuned design. However, it requires the whole closed-loop model, its size grows with the network ($53$ states and $1431$ decision variables per sub-region here), and PI synthesis with it is a nonconvex problem over the gains of all ESSs.
In contrast, each local condition is checked by a small LMI of one ESS model (Rem.~\ref{rem:KYP}), so its size is independent of the network. It is not convex in the gains either, but involves only the gains of one ESS and, for second-order subsystems, is ensured by the closed-form constraints of Prop.~\ref{prop:second-order coefficient condition}.
(ii) The PR-based stability criterion, i.e., Th.~\ref{thm: uniform bound} with $\theta_0=\omega_0=\sigma_0=0$, also enforces separate PR with uniform indices.
In the gain planes of Fig.~\ref{fg:design map} and at the nominal point, it covers $58.2\%$ (boost) and $56.5\%$ (buck) of the stabilizing gains, comparable to the coverages for single sub-regions. However, it ensures stability only: $78.4\%$ (boost) and $71.4\%$ (buck) of the certified gains violate $\calDcap$.}

\section{Conclusion}
\vspace{-3pt}
This paper proposes a decentralized certificate for regional pole placement in linear time-invariant networked systems. 
It combines a complex region mapping and an angle compensation with the compositional property of PR.
The primary theoretical result is the equivalence between the PR of a mapped subsystem's real-equivalent representation and the positivity of its complex-coefficient transfer function. 
Building on this equivalence, the certificate consists of a local positivity condition on each subsystem and a condition on the network matrix, and requires neither model sharing nor communication among the subsystems. 
For DC microgrids, the certificate yields an operating-point-robust grid code that supports plug-and-play operation.

Future research may combine the region mapping and the angle compensation with other compositional criteria, such as integral quadratic constraints. Moreover, the method may extend to power-electronics-interfaced AC grids, whose converters have complex-coefficient transfer functions in the sequence domain.


\bibliographystyle{plain}
\bibliography{autosam}

\end{document}